\documentclass[conference]{IEEEtran}
\usepackage{times,amsmath,color,amssymb,graphicx,epsfig}

\usepackage{amsmath}
\usepackage{amssymb}
\usepackage{amsfonts}
\usepackage{graphicx}
\usepackage{epsfig}

\newtheorem{theorem}{Theorem}
\newtheorem{proposition}{Proposition}
\newtheorem{lemma}{Lemma}

\newcommand{\bbA}{{\bf A}}

\newcommand{\bbB}{{\bf B}}

\newcommand{\bbE}{{\bf E}}

\newcommand{\bbI}{{\bf I}}

\begin{document}

\title{On the Performance of Spectrum Sensing Algorithms using Multiple Antennas}

\author{Ying-Chang Liang$^*$, Guangming Pan$^\dag$, and Yonghong Zeng$^*$\\
* Institute for Infocomm Research, A*STAR, Singapore\\
 $\dag$ School of Physical and Mathematical Sciences, Nanyang
Technological University, Singapore.\\
Emails: ycliang@i2r.a-star.edu.sg, GMPAN@ntu.edu.sg,
yhzeng@i2r.a-star.edu.sg}

\maketitle
\begin{abstract}
In recent years, some spectrum sensing algorithms using multiple
antennas, such as the eigenvalue based detection (EBD), have
attracted a lot of attention. In this paper, we are interested in
deriving the asymptotic distributions of the test statistics of the
EBD algorithms. Two EBD algorithms using sample covariance matrices
are considered: maximum eigenvalue detection (MED) and condition
number detection (CND). The earlier studies usually assume that the
number of antennas $(K)$ and the number of samples $(N)$ are both
large, thus random matrix theory (RMT) can be used to derive the
asymptotic distributions of the maximum and minimum eigenvalues of
the sample covariance matrices. While assuming the number of
antennas being large simplifies the derivations, in practice, the
number of antennas equipped at a single secondary user is usually
small, say $2$ or $3$, and once designed, this antenna number is
fixed. Thus in this paper, our objective is to derive the asymptotic
distributions of the eigenvalues and condition numbers of the sample
covariance matrices for any fixed $K$ but large $N$, from which the
probability of detection and probability of false alarm can be
obtained. The proposed methodology can also be used to analyze the
performance of other EBD algorithms. Finally, computer simulations
are presented to validate the accuracy of the derived results.
\end{abstract}

\begin{keywords}
Spectrum Sensing, Cognitive Radio, Random Matrix Theory.
\end{keywords}

\IEEEpeerreviewmaketitle

\section{Introduction}
\label{sec:intro}

Spectrum sensing is one of the key elements in opportunistic spectrum access design \cite{Haykin2005}, \cite{Liang08}.
In literature, there are a few main categories of practical spectrum
sensing schemes for OSA including\footnote{Matched-filtering method
requires perfect knowledge of the PU's signal received at SU, which
is almost impossible in cognitive radio scenario. Thus
matched-filtering method is not treated a practical spectrum sensing
scheme here.}: energy detection, cyclostationarity based detection,
and eigenvalue-based detection (EBD); see, e.g. \cite{Yonghong10}
and references therein. Energy detection is simple for
implementation, however, it requires accurate noise power
information, and a small error in that estimation may cause SNR wall
and high probability of false alarm \cite{Tandra08}. For
cyclostationarity based detection, the cyclic frequency of PU's
signal needs to be acquired \emph{a priori}. On the other hand, EBD
\cite{Zeng_Liang_TCOM09}, first proposed
and submitted to IEEE 802.22 working group
\cite{Zeng_Liang_802_06-1}, performs signal detection by estimating
the signal and noise powers simultaneously, thus it does not require
the cyclic knowledge of the PU and is robust to noise power
uncertainty.

In general, two steps are needed to complete a sensing design: (1)
to design the test statistics; and (2) to derive the probability
density function (PDF) of the test statistics. For spectrum sensing
with multiple antennas, the test statistics can be designed based on
the standard principles such as generalized likelihood ratio testing
(GLRT) \cite{Zhang_GLRT}, or other considerations \cite{Zeng_Liang_TCOM09}, \cite{Zeng_Koh_Liang_ICC2008},
\cite{Cardoso_Debbah_ISWPC08},
\cite{Penna_Garello_COMLetter09}.
Surprisingly these studies all give the test statistics using the
eigenvalues of the sample covariance matrix. It is thus important to derive the PDF of the eigenvalues so that the sensing performance can be quantified.

For EBD schemes, the PDF of the test statistics is usually
derived using random matrix theory (RMT); see, e.g., \cite{Zeng_Liang_TCOM09}, \cite{Zeng_Koh_Liang_ICC2008}, \cite{Cardoso_Debbah_ISWPC08},
\cite{Penna_Garello_COMLetter09}. In fact, the maximum and minimum eigenvalues of the sample covariance matrix
have simple explicit expressions when both antenna number ($K$) and ample size ($N$) are large
\cite{Johnstone2001,Bai_Fang_Liang_book}.
It is reasonable to assume that $N$ is large especially when the
secondary user is required to sense a weak primary signal, in
practice, however, the number of antennas equipped at a single
secondary user is usually small, say $2$ or $3$. Thus the results
obtained under the assumption that both $K$ and $N$ are large may
not be accurate for practical multi-antenna cognitive radios. In
this paper, our objective is to derive the asymptotic distributions
of the eigenvalues of the sample covariance matrices for arbitrary
$K$ but large $N$. The asymptotic results obtained form the basis
for quantifying the PDF of the test statistics for EBD algorithms.

It is noticed that there are studies on the exact distribution of
the condition number of sample covariance matrices for arbitrary $K$
and $N$ \cite{Penna_CrownCom09}, the formulas derived however are
complex and cannot be conveniently used to analyze the sensing
performance. Furthermore, there are no results published for the
case when the primary signals exist.

The rest of the paper is organized as follows.
Section~\ref{sec:spectrum sensing} presents the system model for
spectrum sensing using multiple antennas. Two EBD algorithms are
reviewed in Section \ref{sec:EBD}, including maximum eigenvalue
detection (MED) and condition number detection (CND) algorithms. In
Section \ref{sec:asym-1}, we derive the asymptotic distributions of
the test statistics of the two EBD algorithms for the scenario when
the primary users are inactive. In Section \ref{sec:asym-2}, the
results are derived for the scenario when there are active primary
users in the sensed band. Performance evaluations are given in
Section \ref{sec:evaluations}, and finally, conclusions are drawn in
Section \ref{sec:conclusions}.

The following notations are used in this paper. Matrices and vectors
are typefaced using bold uppercase and lowercase letters,
respectively. Conjugate transpose of matrix $\bbA$ is denoted as
$\bbA^{H}$. $\| {\bf a} \|$ stands for the norm of vector ${\bf a}$.
$\mathbb{E}[\cdot]$ denotes expectation operation; $\stackrel{D}
{\longrightarrow}$, $\stackrel{a.s.} {\longrightarrow}$ and
$\stackrel{\Delta}{=}$ stand for ``convergence in distribution'',
``almost surely convergence''; and ``defined as'', respectively.

\section{System Model}
\label{sec:spectrum sensing}

Suppose the SU is equipped with $K$ antennas, which are all used to
sense a radio spectrum for collecting $N$ samples each. One of the following two scenarios
happens.

{\bf Scenario 1 ($\mathcal{S}_{1}$):} There are $t~(t \geq 1)$
active PU transmissions in the sensed band, and the sampled outputs,
${\bf x}(n)$, $n=0,\cdots,N-1$, can be represented as
\begin{eqnarray}
{\bf x}(n) = {\bf H} {\bf s} (n) + {\bf u} (n), \label{eq:H1}
\end{eqnarray}
where ${\bf u}(n)$ is the $(K \times 1)$ noise vector, ${\bf s}(n)$
is the $(t \times 1)$ signal vector containing the transmitted
signals from the active PUs, and ${\bf H} = [{\bf h}_{1}, \cdots,
{\bf h}_{t} ]$ is the $(K \times t)$ channel matrix from the active
PUs to the SU.

{\bf Scenario 0 ($\mathcal{S}_{0}$):} There are no
active PUs in the sensed band, thus the sampled outputs collected at
SU are given by
\begin{eqnarray}
{\bf x} (n) = {\bf u}(n). \label{eq:H0}
\end{eqnarray}

We make the following assumptions for the above models:

\begin{itemize}

\item[(i)] The noises ${\bf u}(n)$ are independent and
identically distributed (iid) both spatially and temporally. Each
element follows Gaussian distribution with mean zero and variance
$\sigma_{u}^{2}$.

\item[(ii)] For each $n$, the elements of ${\bf s}(n)$ are iid, and follow Gaussian distribution with mean zero and variance $\sigma_{s}^{2}$. Thus the covariance matrix of ${\bf s}(n)$ is ${\bf R}_{s} \stackrel{\Delta}{=} \mathbb{E} [{\bf s}(n) {\bf s}^{H}(n)] = \sigma_{s}^{2} {\bf I}$.

\item[(iii)] ${\bf s}(n)$ are ${\bf u}(n)$ are independent of each other.

\item[(iv)] The eigenvalues of ${\bf H} {\bf H}^{H}$ are not all identical.

\end{itemize}

Two cases of Gaussian distribution are considered: (1)
\emph{Real-valued case:} we denote ${\bf u}(n) \sim \mathcal{N}
({\bf 0}, \sigma_{u}^{2} {\bf I})$ and ${\bf s}(n) \sim \mathcal{N}
({\bf 0}, \sigma_{s}^{2} {\bf I})$; (2) \emph{Complex-valued case:}
We denote ${\bf u}(n) \sim \mathcal{CN} ({\bf 0}, \sigma_{u}^{2}
{\bf I})$ and ${\bf s}(n) \sim \mathcal{CN} ({\bf 0}, \sigma_{s}^{2}
{\bf I})$;

When there are active PUs, we define the average received
signal-to-noise ratio (SNR) of PUs' signals measured at SU as
: $\mbox{SNR} = \frac{\sum_{\ell=1}^{t} \| {\bf h}_{\ell} \|^{2}
\sigma_{s}^{2}}{K \sigma_{u}^{2}}$. Given the received samples, ${\bf X} = [{\bf
x}(0), {\bf x}(1), \cdots, {\bf x}(N-1)]$, the objective of spectrum sensing design is to choose one of the two
hypothesis: $\mathcal{H}_{0}$: there are no active primary users in
the sensed band, and $\mathcal{H}_{1}$: there exist active primary
users. For that, we need to design the test static, $T( {\bf X})$,
and a threshold $\epsilon$, and infer $\mathcal{H}_{1}$ if $T( {\bf
X}) > \epsilon$, and infer $\mathcal{H}_{0}$ if $T( {\bf X}) \leq
\epsilon$.

\section{Eigenvalue based Detections}
\label{sec:EBD}

Let us define the \emph{sample covariance matrix} and
\emph{covariance matrix} of the measurements ${\bf x}(n)$ of the
sensed band as
\begin{eqnarray}
\hat{\bf R}_{x} & \stackrel{\Delta}{=} & \frac{1}{N} {\bf X} {\bf X}^{H} = \frac{1}{N}\sum_{n=0}^{N-1} {\bf x}(n) {\bf x}^{H} (n), \label{eq:Rmx} \\
{\bf R}_{x} & \stackrel{\Delta}{=} & \mathbb{E} [{\bf x}(n) {\bf
x}^{H} (n)], \label{eq:Rmx1}
\end{eqnarray}
respectively. For a fixed $K$, under Scenario $\mathcal{S}_{0}$ and
when $N \to \infty$, we have
\begin{eqnarray}
\hat{\bf R}_{x} \to {\bf R}_{x} = \sigma_{u}^{2} {\bf I},
\end{eqnarray}
which means that all the eigenvalues of the \emph{covariance matrix}
are equal to $\sigma_{u}^{2}$. However, under $\mathcal{S}_{1}$ and
when $N \to \infty$, $\hat{\bf R}_{x}$ approaches
\begin{eqnarray}
{\bf R}_{x} = \sigma_{s}^{2} {\bf H} {\bf H}^{H} + \sigma_{u}^{2}
{\bf I}. \label{eq:Rx1}
\end{eqnarray}
Based on assumption (iv), the eigenvalues of \eqref{eq:Rx1} can be
ordered as $\rho_{1} \geq \rho_{2} \geq \cdots \geq \rho_{K}
\geq \sigma_{u}^{2}$.

Denote the ordered eigenvalues of the sample covariance matrix
$\hat{\bf R}_{x}$ as $\hat{\lambda}_{1} \geq \hat{\lambda}_{2} \geq
\cdots \geq \hat{\lambda}_{K}$. We consider the following two EBD algorithms.

\begin{itemize}

\item[(1)] {\bf Maximum Eigenvalue Detection (MED)}:
For MED, the test static is chosen as \cite{Zeng_Koh_Liang_ICC2008}:
\begin{eqnarray}
T^{(m)}({\bf X}) = \frac{\hat{\lambda}_{1}}{\sigma_{u}^{2}}.
\end{eqnarray}

\item[(2)] {\bf Condition Number Detection (CND)}: The CND chooses the following test statistic \cite{Zeng_Liang_TCOM09}:
\begin{eqnarray}
T^{(c)}({\bf X}) = \frac{\hat{\lambda}_{1}}{\hat{\lambda}_{K}}.
\end{eqnarray}
\end{itemize}

An essential task to complete the sensing design is to determine the
test threshold, which affects both probability of detection and
probability of false alarm. To do so, it is important to derive the
PDF of the test statistics.

For arbitrary $(K,N)$ pair, the closed-form expressions of the PDF
of the test statics are in general complex \cite{Penna_CrownCom09}.
In practice, the primary users need to be detectable in low SNR
environment. For example, in IEEE 802.22, the TV signal needs to be
detected at $-20$dB SNR with $90 \%$ target probability of
detection and $10 \%$ target probability of
false alarm. To achieve that, the number of samples, $N$, required for
spectrum sensing is usually very large. In the paper, we thus turn
our attention to derive the PDFs of the test statistics for any
fixed $K$, but large $N$. The asymptotic distributions (when $N \to
\infty$) of the test statics of the EBD algorithms will be derived
for $\mathcal{S}_{0}$ and $\mathcal{S}_{1}$ scenarios, respectively.

\section{Asymptotic Distributions of the Test Statistics under Scenario $\mathcal{S}_{0}$}
\label{sec:asym-1}

Let $\lambda_{i} = \frac{\hat{\lambda}_{i}}{\sigma_{u}^2}$ for
$i=1,\cdots,K$, and denote $\bbA = \mbox{diag} \left \{ \lambda_{1},
\lambda_{2}, \cdots, \lambda_{K} \right \}$. Define $\beta_{i} =
\sqrt{N} (\lambda_i -1)$ for $i=1,\cdots,K$, and $\bbB = \mbox{diag}
\left \{\beta_{1}, \beta_{2}, \cdots, \beta_{K} \right \}$. Note
that $\beta_{1} \geq \beta_{2} \geq \cdots \geq \beta_{K}$.

Under $\mathcal{S}_{0}$, by Theorem 1 and (2.12) in \cite{an1}, we
have the following proposition.

\begin{proposition}
\label{Prop:P1} For real-valued case, when $N \to \infty$, the limiting distribution of
$\bbB$, $g_K(\beta_{1}, \beta_{2}, \cdots, \beta_{K})$, is given by
\begin{eqnarray}
&& g_K(\beta_1,\cdots,\beta_K)\nonumber \\
& = &C_1(K) \exp \left (-\frac{1}{4}\sum\limits_{i=1}^K \beta_i^2
\right )\prod_{1 \leq i<j \leq K}(\beta_i-\beta_j), \label{eq:a1}
\end{eqnarray}
where
\begin{eqnarray}
C_1(K)=\frac{1}{2^{K(K+3)/4}\prod_{i=1}^K\Gamma[\frac{1}{2}(K+1-i)]}.
\end{eqnarray}
Here $\Gamma(z) = \int_{0}^{\infty} t^{z-1} e^{-t} dt$ is the Gamma function.
\end{proposition}

For real-valued case, and when $K = 1$, we have $\beta_{1} \sim \mathcal{N}
(0,2)$, i.e., $g_{1}(\beta_{1}) =\frac{1}{2 \sqrt{\pi}} e^{-\frac{\beta_{1}^2}{4}}$. When $K = 2$, we have
\begin{equation}
g_2(\beta_1,\beta_2)=\frac{1}{2^{5/2}\sqrt{\pi}}(\beta_1-\beta_2)e^{-\frac{\beta_1^2+\beta_2^2}{4}},\quad
\beta_1\geq \beta_2.
\end{equation}

For \emph{complex-valued case}, by Lindberg's central limit theorem
\cite{Petrov_book}, $\sqrt{N} (\frac{\hat{\bf
R}_{x}}{\sigma_{u}^{2}} - {\bf I})$ converges to a Hermitian matrix
with elements above diagonal being complex Gaussian distribution, a
so-called Gaussian unitary matrix. Thus, according to the joint
density of the eigenvalues of a Gaussian unitary matrix, similar to
Theorem 1 in \cite{an1} we have the following:

\begin{proposition}
\label{Prop:P2} For complex-valued case, when $N \to \infty$, the limiting distributions
of $\bbB$, $g_K(\beta_{1}, \beta_{2}, \cdots, \beta_{K})$, is given
by
\begin{eqnarray}
&& g_K(\beta_1,\cdots,\beta_K) \nonumber \\
& = & C_2(K) \exp \left (-\frac{1}{2}\sum\limits_{i=1}^K\beta_i^2
\right )\prod_{1 \leq i<j \leq K}(\beta_i-\beta_j)^2, \label{eq:a2}
\end{eqnarray}
where
\begin{eqnarray}
C_2(K)=K!(2\pi)^{-K/2}\prod_{j=1}^K\frac{1}{\Gamma(1+j)}.
\end{eqnarray}
\end{proposition}

For complex-valued case, when $K = 1$, we have $g(\beta_{1})=\frac{1}{\sqrt{2 \pi}}
e^{-\frac{\beta_{1}^2}{2}}$, i.e., $\beta_{1} \sim \mathcal{N}
(0,1)$. When $K=2$, we have
\begin{equation}
g_2(\beta_1,\beta_2)=\frac{1}{2\pi}(\beta_1-\beta_2)^2e^{-\frac{\beta_1^2+\beta_2^2}{2}},\quad
\beta_1\geq \beta_2.
\end{equation}

For a given $K$, the limiting distribution, $\bar{g}_{K,i}(x)$, of $\beta_{i}, i=1,\cdots,K$, can be calculated from \eqref{eq:a1} or \eqref{eq:a2}.

\subsection{MED}

For MED, $T^{(m)}({\bf X}) =
\frac{\hat{\lambda}_{1}}{\sigma_{u}^{2}} = \lambda_{1}$, thus
$\sqrt{N} (T^{(m)}({\bf X}) -1) = \beta_{1}$. We have the following:

\begin{theorem}
Under $\mathcal{S}_{0}$, for a fixed $K$ and when $N \to \infty$,
\begin{equation} \label{a1} \sqrt{N} \left ( T^{(m)}({\bf X}) -1
\right )\stackrel{D}\longrightarrow m_{K}^{(0)},
\end{equation}
whose distribution, $f_{K}^{(m)}(x)$, is the same as the limiting distribution of
$\beta_{1}$, $\bar{g}_{K,1}(x)$.
\end{theorem}

\subsection{CND}

For CND, $T^{(c)}({\bf X}) =
\frac{\hat{\lambda}_{1}}{\hat{\lambda}_{K}} =
\frac{\lambda_{1}}{\lambda_{K}}$. The following theorem states the
limiting distribution of the condition number of the sample
covariance matrix.

\begin{theorem}
Under $\mathcal{S}_{0}$, for a fixed $K$ and when $N \to \infty$,
\begin{equation} \label{a1} \sqrt{N} \left (T^{(c)}({\bf X}) -1
\right )\stackrel{D}\longrightarrow z_{K}^{(0)},
\end{equation}
whose PDF and cumulative density function (CDF) are given by
\begin{eqnarray}
f_K^{(c)}(x) & = & \int^{\infty}_{-\infty} \tilde{g}_K(\beta_K+
x,\beta_K)d\beta_K, \label{eq:PDF_K}\\
F_K^{(c)}(x) & = & \int^{\infty}_{-\infty}\int_{\beta_K}^{\beta_K+
x} \tilde{g}_K(\beta_1,\beta_K)d\beta_1d\beta_K, \label{eq:CDF_K}
\end{eqnarray}
respectively, where $\tilde{g}_K(\beta_1,\beta_K)$ is the joint
distribution of $\beta_{1}$ and $\beta_{K}$.
\end{theorem}

\begin{proof}
See \cite{Liang_Pan_zeng10}.
\end{proof}

Next, we derive the closed-form expressions of the limiting
distributions of the condition number of the sample covariance
matrix for $K = 2$ case.

\begin{lemma}
\label{lemma-1} Consider Scenario $\mathcal{S}_{0}$. For $K=2$ and
$N \to \infty$,
\begin{equation}
\label{a1} \sqrt{N} \left ( T^{(c)}({\bf X}) -1 \right
)\stackrel{D}\longrightarrow z_{2}^{(0)},
\end{equation}
whose PDF is
\begin{eqnarray}
f_{2}^{(c)}(x) = \frac{1}{4}x \exp(-x^2/8),\qquad x\geq 0,
\end{eqnarray}
for real-valued case, and
\begin{eqnarray}
f_{2}^{(c)}(x)=\frac{1}{2\sqrt{\pi}}x^2e^{-x^2/4},\qquad x\geq 0.
\end{eqnarray}
for complex-valued case.
\end{lemma}

\subsection{Threshold Determination}

From the above two subsections, it is seen that for a given $K$ and
when $N \to \infty$, the \emph{regulated test statistics},
$\sqrt{N}(T^{(i)} ({\bf X}) -1)$ converges to a random variable with
PDF $f_{K}^{(i)}(x)$ and CDF $F_{K}^{(i)}(x)$, where $i=m$ for MED and $i=c$ for CND.
Here, we look at how to set the decision threshold $\epsilon$ to
achieve a target probability of false alarm, $\bar{P}_{f}$. Since
\begin{eqnarray}
\bar{P}_{f} & = & \mbox{Prob}~\left ( T^{(i)}({\bf X}) > \epsilon | \mathcal{S}_{0} \right ) \nonumber \\
& = & \mbox{Prob}~\left ( \sqrt{N} \left ( T^{(i)}({\bf X})-1 \right ) > \sqrt{N} (\epsilon -1) | \mathcal{S}_{0} \right ) \nonumber \\
& = & 1 - \int_{0}^{\sqrt{N} (\epsilon -1)} f_{K}^{(i)} (t) dt \nonumber \\
& = & 1 - F_{K}^{(i)} (\sqrt{N} (\epsilon -1)).
\end{eqnarray}
Thus the decision threshold is determined by
\begin{eqnarray}
\epsilon = 1 + \frac{1}{\sqrt{N}} [F_{K}^{(i)}]^{-1}(1-\bar{P}_{f}),
\end{eqnarray}
where $x=[F_{K}^{(i)}]^{-1}(y)$ denotes the inverse function of
$y=F_{K}^{(i)} (x)$.

\section{Asymptotic Distributions of the Test Statistics under Scenario $\mathcal{S}_{1}$}
\label{sec:asym-2}

In this section, we derive the asymptotic distributions of the test
statistics under $\mathcal{S}_{1}$. Consider the ordered eigenvalues
of covariance matrix ${\bf R}_{x}$, $\rho_{1} \geq \rho_{2} \geq
\cdots \geq \rho_{K}$, and the ordered eigenvalues of sample
covariance matrix $\hat{{\bf R}}_{x}$, $\hat{\lambda}_{1} \geq
\hat{\lambda}_{2} \geq \cdots \geq \hat{\lambda}_{K}$. With the same notation as in \cite{an1}, let the multiplicities of the
eigenvalues of ${\bf R}_{x}$ be $q_1,q_2,\cdots,q_r$. That is
\begin{eqnarray}
&& \rho_1=\cdots=\rho_{q_1}=\mu_1,\nonumber\\
&& \rho_{q_1+1}=\cdots=\rho_{q_{1}+q_2-1}=\mu_2, \nonumber\\
&& \rho_{K-q_r+1}=\cdots=\rho_{K}=\mu_r,\nonumber
\end{eqnarray}
where
$$
\mu_1>\mu_{2}>\cdots>\mu_{r}.
$$

Moreover, let
\begin{eqnarray}
&& {\bf{\Gamma}}^H {\bf{R}}_{x} {\bf{\Gamma}} = {\bf{\Delta}},\nonumber \\
&& \bbE^H {\bf{\Gamma}}^H \hat{\bf{R}}_{x} {\bf{\Gamma}} \bbE =
\bbB,\nonumber
\end{eqnarray}
where $${\bf{\Delta}}=\mbox{diag}(\rho_1,\rho_2,\cdots,\rho_K),$$
$$\bbB=\mbox{diag}(\hat{\lambda}_1,\hat{\lambda}_2,\cdots,\hat{\lambda}_K),$$ and
$\bf{\Gamma}$ and $\bbE$ are corresponding eigenvector matrices. We
then partition the matrices ${\bf{\Delta}}$ and $\bbB$ into block
matrices as follows:
\begin{eqnarray}
\bf{\Delta}=\begin{pmatrix}
 \mu_1\bbI& 0& \cdots &0\\
         0&\mu_2\bbI&\cdots &0\\
          \vdots&\vdots&\ &\vdots\\
          0&0&\cdots &\mu_r\bbI
\end{pmatrix},
\end{eqnarray}
\begin{eqnarray}
\bbB=\begin{pmatrix}
 \bbB_1& 0& \cdots &0\\
         0&\bbB_2&\cdots &0\\
          \vdots&\vdots&\ &\vdots\\
          0&0&\cdots &\bbB_r
\end{pmatrix}.
\end{eqnarray}
Note in both ${\bf{\Delta}}$ and $\bbB$, the $k$th diagonal
sub-block is $q_{k} \times q_{k}$ matrix.

Following the proofs of Proposition 1 and Proposition 2, and Theorem
1 of \cite{an1}, we have the following results.

\begin{theorem}
\label{theorem:H1}
The limiting
distribution of $\sqrt{N} \left ( \frac{\bbB_k}{\mu_k}-\bbI \right)$
is given by \eqref{eq:a1} and \eqref{eq:a2} for real-valued case and
complex-valued case, respectively, with the parameter $K$ being
replaced by $q_{k}$. Furthermore, $\sqrt{N} \left (
\frac{\bbB_i}{\mu_i}-\bbI \right)$  is asymptotically independent of
$\sqrt{N} \left ( \frac{\bbB_j}{\mu_j}-\bbI \right)$  for $i \neq
j$.
\end{theorem}

Let $\gamma_{1} = \sqrt{N} (
\frac{\hat{\lambda}_{1}}{\mu_{1}} -1)$ and $\gamma_{K} =
\sqrt{N}(\frac{\hat{\lambda}_{K}}{\mu_{r}} -1)$. From Theorem \ref{theorem:H1}, $\gamma_{1}$ and $\gamma_{K}$ are asymptotically independent, and the limiting
distributions of $\gamma_{1}$ and $\gamma_{K}$ are $\bar{g}_{q_{1},1}(x)$ and $\bar{g}_{q_{r},q_{r}}(x)$, respectively, which are defined in Section \ref{sec:asym-1}.

\subsection{MED}
For MED, $T^{(m)}({\bf X}) =
\frac{\hat{\lambda}_{1}}{\sigma_{u}^{2}}$, thus
$\frac{\hat{\lambda}_{1}}{\mu_1} = \alpha^{(m)} T^{(m)}({\bf X})$ where
$\alpha^{(m)} = \frac{\sigma_{u}^{2}}{\mu_{1}}$. We have the following
theorem.
\begin{theorem}
Under $\mathcal{S}_{1}$, when $N \to \infty$,
\begin{eqnarray}
\sqrt{N} \left ( \alpha^{(m)} T^{(m)}({\bf X}) -1 \right )
\stackrel{D}\longrightarrow m_{K}^{(1)},
\end{eqnarray}
whose PDF is $\bar{g}_{q_{1},1}(x)$.
\end{theorem}

If there are no repeated maximum eigenvalues (i.e., $q_{1} = 1$), $m_{K}^{(1)} \sim \mathcal{N}
(0,2)$ for real-valued case, and $m_{K}^{(1)} \sim \mathcal{N}
(0,1)$ for complex-valued case.

\subsection{CND}

For CND, $T^{(c)}({\bf X}) = \frac{\hat{\lambda}_{1}}{\hat{\lambda}_{K}}$.
Since $\gamma_{1}$ and $\gamma_{K}$ are asymptotically independent,
and when $N \to \infty$, $\frac{\hat{\lambda}_{1}}{\mu_{1}}
\stackrel{a.s.}\longrightarrow 1$,
$\frac{\hat{\lambda}_{K}}{\mu_{r}}
\stackrel{a.s.}\longrightarrow 1$, similar to Theorem 2, we have the
following theorem.
\begin{theorem}
Under $\mathcal{S}_{1}$, for a fixed $K$ and when $N \to \infty$,
\begin{eqnarray}
\sqrt{N} \left ( \alpha^{(c)} T^{(c)}({\bf X}) -1 \right
)\stackrel{D}\longrightarrow c_K^{(1)},
\end{eqnarray}
where $\alpha^{(c)} = \frac{\mu_{r}}{\mu_{1}}$, and the PDF of $c_K^{(1)}$ is given by
\begin{eqnarray}
w_K^{(c)} (x) & = & \int^{\infty}_{-\infty} \bar{g}_{1,1}(x+y)
\bar{g}_{q_{r},q_{r}}(y) dy.
\end{eqnarray}
\end{theorem}

Now let us look at the case when $q_{1} = q_{r} = 1$. Since both $\gamma_{1}$ and $\gamma_{K}$ follow the same Gaussian distribution asymptotically, we can easily derive that
$c_{K}^{(1)} \sim \mathcal{N}
(0,4)$ for real-valued case, and $c_{K}^{(1)} \sim \mathcal{N}
(0,2)$ for complex-valued case.

\section{Computer Simulations}
\label{sec:evaluations}

Computer simulations are presented in this section to validate the
effectiveness of the results obtained in this paper. We choose $K=2$, and there is one primary signal under $\mathcal{S}_{1}$.
$10000$ Monte Carlo runs are carried out in order to compute CDF of
the test statistics.

\subsection{Scenario $\mathcal{S}_{0}$}

We first compare the CDFs of the regulated test statistics
$\sqrt{N} (T^{(m)}( {\bf X}) -1)$ for MED derived using computer
simulations, and theoretic analysis under fixed $K$ assumption of this paper and large $K$
assumption of \cite{Johnstone2001}. The results for real-valued case with $N = 1000$ are shown in Fig. \ref{fig:fig1}. It is seen that
the simulated CDF and theoretical CDF derived in this paper are
close to each other, while the result derived under large $K$
assumption are far from the simulated one.

\begin{figure}
    \centering
        \includegraphics[width = 0.5\textwidth]{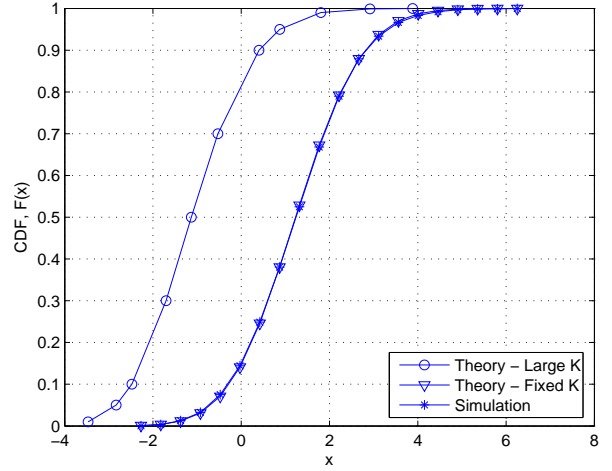}
    \caption{Simulated and theoretical CDFs of the regulated test statistics under $\mathcal{S}_{0}$: Real-valued case, MED with $N=1000$.}
    \label{fig:fig1}
\end{figure}

We next evaluate the accuracy of threshold setting for the CND
detection using the formula in Section IV. Fig. \ref{threshold1}
shows the threshold values at different probability of false alarms
for real-valued case with $N=10000$. For comparison, the true thresholds (based on
simulations) and the thresholds by using the theory for large $K$ in
\cite{Zeng_Liang_TCOM09} are also included in the figure. The
proposed theory based on fixed $K$ is much more accurate than the
theory in \cite{Zeng_Liang_TCOM09} and gives the threshold values
approaching to the true values.

\begin{figure}
\centering
\includegraphics[width=8cm]{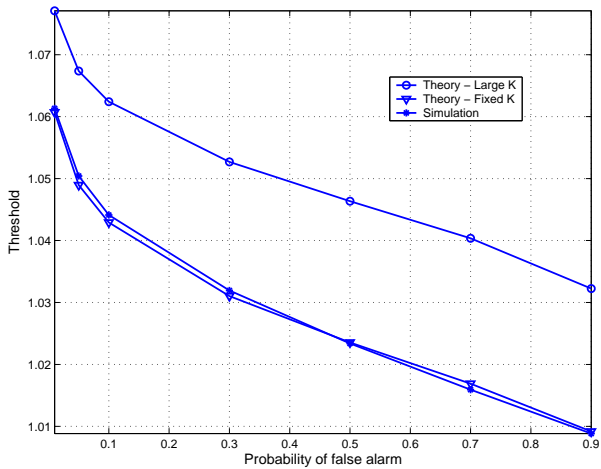}
\caption{Comparison of threshold values under $\mathcal{S}_{0}$: Real-valued case, CND with $N=10000$.}\label{threshold1}
\end{figure}

\subsection{Scenario $\mathcal{S}_{1}$}

We then compare the detection probability results predicted by the formulas derived in this paper and reference \cite{Zeng_Liang_TCOM09} with large $K$ assumption.
The probabilities of detection for CND by using simulations and different
theoretical formulas are shown in Fig. \ref{pd1} with $N=10000$ and SNR = $-15$dB (real-valued case). For a target probability of false alarm, we first use simulations to determine the decision threshold, then obtain the probability of detection based on simulations or the related formulas. Again, the
proposed formula based on fixed $K$ is much more accurate than the
formula in \cite{Zeng_Liang_TCOM09} and gives the values approaching
to the true values. From Fig. \ref{pd1}, it is also seen that, interestingly, for a target probability of false alarm, the probability of detection predicted by our method tends to be conservative, which seems to be good to primary users in terms of protection requirement. Finally, as the results derived in this paper are asymptotic, they will approach the simulated ones more accurately when the sample size increases.

\begin{figure}
\centering
\includegraphics[width=8cm]{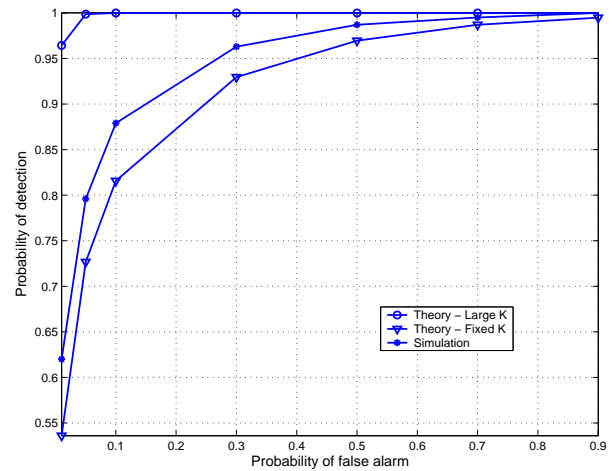}
\caption{Comparison of estimations for probability of detection
under $\mathcal{S}_{1}$: Real-valued case, CND with $N=10000$ and SNR = $-15$dB.}\label{pd1}
\end{figure}

\section{Conclusions}
\label{sec:conclusions} In this paper, theoretic distributions of
the test statistics for some eigenvalue based detections have been
derived for any fixed $K$ but large $N$, from which the probability
of detection and probability of false alarm have been obtained. The
results are useful in practice to determine the detection thresholds
and predict the performances of the sensing algorithms. Extensive simulations have shown that
theoretic results have higher accuracy than existing stuties.

\end{document}